\PassOptionsToPackage{x11names,svgnames,usenames,dvipsnames,table,rgb}{xcolor}

\documentclass[manuscript,screen,nonacm,prologue,x11names,svgnames,usenames,dvipsnames,table,rgb]{acmart}

\usepackage[T1]{fontenc}
\usepackage{graphicx}
\usepackage{hyperref}
\usepackage{enumerate}

\usepackage{adjustbox}
\usepackage{physics}

\usepackage{algorithm,algorithmic}
\usepackage{subfig}
\usepackage{array}
\newcolumntype{H}{>{\setbox0=\hbox\bgroup}c<{\egroup}@{}}

\usepackage[x11names,svgnames,usenames,dvipsnames,table,rgb]{xcolor}

\usepackage{tikz}
\usetikzlibrary{arrows}
\usetikzlibrary{positioning}  
\usetikzlibrary{shadows}
\usetikzlibrary{arrows,backgrounds,calc,chains,
	matrix,positioning,shapes,shapes.geometric,
	shapes.arrows, decorations.pathmorphing, decorations.pathreplacing}

\tikzstyle{vecArrow} = [thick, decoration={markings,mark=at position
	1 with {\arrow[semithick]{open triangle 60}}}]
 
\usetikzlibrary{positioning,calc,fit,backgrounds,matrix,shapes,arrows.meta,graphs,quotes,automata,shapes.geometric,chains}

\tikzset{%
	font={\footnotesize},
	vertex/.style={draw,circle,inner sep=0pt,minimum width=0.5cm,minimum height=0.5cm,font=\small, scale=0.9},
	terminal/.style={draw,fill=white,rectangle,inner sep=2pt,font=\footnotesize,very thick},
	define color/.code={\definecolor{hsb#1}{Hsb}{#1, 1, 0.75}},
	medge/.style n args={3}{
			line width={#1pt},
			define color={#2},
			draw=hsb#2,
			out=#3,
			in=90
		},
	edge/.style 2 args={
			line width={#1pt},
			define color={#2},
			draw=hsb#2
		},
	edge0/.style 2 args={
			line width={#1pt},
			define color={#2},
			draw=hsb#2,
			out=-130,
			in=90
		},
	edge1/.style 2 args={
			line width={#1pt},
			define color={#2},
			draw=hsb#2,
			out=-50,
			in=90
		},
	zerostub/.style={
			inner sep=0,
			minimum size=3pt,
			circle,
			fill=black
		},
	cross/.style={cross out, draw=black,
			minimum size=2*(#1-\pgflinewidth),
			inner sep=0pt, outer sep=0pt,rotate=45},
	cross/.default={3pt},
	edgeOne/.style={color=RedEdge,ultra thick},
	edgeOneState/.style={color=RedEdge, thick},
	edgeMOne/.style={color=BlueEdge,ultra thick},
	edgeSqrt/.style={color=RedEdge, thick},
	edgeMSqrt/.style={color=BlueEdge, thick},
	edgeFrac/.style={color=RedEdge, thin},
	edgeOver/.style={dotted, color=blue, ultra thick},
	qubit/.style={draw,circle,inner sep=0pt,minimum width=0.35cm,minimum height=0.35cm,font=\footnotesize, thin}
}
\definecolor{RedEdge}{RGB}{191,40,40}
\definecolor{BlueEdge}{RGB}{40,191,191}
\definecolor{Blue01}{rgb}{0.26,0.39,0.85}
\definecolor{Yellow12}{rgb}{1.0,0.88,0.1}
\definecolor{Gray23}{rgb}{0.66,0.66,0.66}
\definecolor{myyellow}{RGB}{255, 253, 0}
\definecolor{myorange}{RGB}{238, 135, 51}
\definecolor{myblue}{RGB}{47, 112, 137}

\usepackage{multirow}
\usepackage{wrapfig}

\AtEndPreamble{%
    \theoremstyle{acmdefinition}
    \newtheorem{remark}[theorem]{Remark}
}

\newcommand{\calO}{\mathcal{O}}

\newcommand{\CO}{\mathbb{C}}

\newcommand{\hU}{\hat{U}}

\AtBeginDocument{%
  \providecommand\BibTeX{{%
    \normalfont B\kern-0.5em{\scshape i\kern-0.25em b}\kern-0.8em\TeX}}}

\setcopyright{acmlicensed}
\copyrightyear{2024}
\acmYear{2024}
\acmDOI{XXXXXXX.XXXXXXX}

\begin{document}

\title{Forward and Backward Constrained Bisimulations for Quantum Circuits using Decision Diagrams}


\author{L.~Burgholzer}
\email{lukas.burgholzer@tum.de}
\orcid{0000-0003-4699-1316}
\affiliation{%
  \institution{Technical University of Munich}
  \streetaddress{Arcisstraße 21}
  \city{Munich}
  \country{Germany}}
  
\author{A.~Jiménez-Pastor}
\email{ajpa@cs.aau.dk}
\orcid{0000-0002-6096-0623}
\affiliation{%
  \institution{Aalborg University}
  \streetaddress{Selma Lagerløfs Vej 300}
  \city{Aalborg}
  \state{Jutland}
  \country{Denmark}
  \postcode{9220}
}
\author{K.G.~Larsen}
\orcid{0000-0002-5953-3384}
\affiliation{%
  \institution{Aalborg University}
  \streetaddress{Selma Lagerløfs Vej 300}
  \city{Aalborg}
  \state{Jutland}
  \country{Denmark}
  \postcode{9220}
}

\author{M.~Tribastone}
\email{mirco.tribastone@imtlucca.it}
\orcid{0000-0002-6018-5989}
\affiliation{%
  \institution{IMT Lucca}
  \streetaddress{...}
  \city{Lucca}
  \country{Italy}}

\author{M.~Tschaikowski}
\email{tschaikowski@cs.aau.dk}
\orcid{0000-0002-6186-8669}
\affiliation{%
  \institution{Aalborg University}
  \streetaddress{Selma Lagerløfs Vej 300}
  \city{Aalborg}
  \state{Jutland}
  \country{Denmark}
  \postcode{9220}
}

\author{R.~Wille}
\email{robert.wille@tum.de}
\orcid{0000-0002-4993-7860}
\affiliation{%
  \institution{Technical University of Munich}
  \streetaddress{Arcisstraße 21}
  \city{Munich}
  \country{Germany}}
\affiliation{%
\institution{Software Competence Center Hagenberg GmbH}
\streetaddress{Softwarepark 32a}
\city{Hagenberg}
\country{Austria}}


\begin{abstract}
Efficient methods for the simulation of quantum circuits on classic computers are crucial for their analysis due to the exponential growth of the problem size with the number of qubits. Here we study lumping methods based on bisimulation, an established class of techniques that has been proven successful for (classic) stochastic and deterministic systems such as Markov chains and ordinary differential equations. Forward constrained bisimulation yields a lower-dimensional model which exactly preserves quantum measurements projected on a linear subspace of interest. Backward constrained bisimulation gives a reduction that is valid on a subspace containing the circuit input, from which the circuit result can be fully recovered. We provide an algorithm to compute the constraint bisimulations yielding coarsest reductions in both cases, using a duality result relating the two notions. As applications, we provide theoretical bounds on the size of the reduced state space for well-known quantum algorithms for search, optimization, and factorization. Using a prototype implementation, we report significant reductions on a set of benchmarks. Furthermore, we show that constraint bisimulation complements state-of-the-art methods for the simulation of quantum circuits based on decision diagrams.
\end{abstract}





\begin{abstract}
Efficient methods for the simulation of quantum circuits on classic computers are crucial for their analysis due to the exponential growth of the problem size with the number of qubits. Here we study lumping methods based on bisimulation, an established class of techniques that has been proven successful for (classic) stochastic and deterministic systems such as Markov chains and ordinary differential equations. Forward constrained bisimulation yields a lower-dimensional model which exactly preserves quantum measurements projected on a linear subspace of interest. Backward constrained bisimulation gives a reduction that is valid on a subspace containing the circuit input, from which the circuit result can be fully recovered. We provide an algorithm to compute the constraint bisimulations yielding coarsest reductions in both cases, using a duality result relating the two notions. As applications, we provide theoretical bounds on the size of the reduced state space for well-known quantum algorithms for search, optimization, and factorization. Using a prototype implementation, we report significant reductions on a set of benchmarks. In particular, we show that constrained bisimulation can boost decision-diagram-based quantum circuit simulation by several orders of magnitude, allowing thus for substantial synergy effects. 
\end{abstract}

\maketitle

\section{Introduction}\label{sec:intro}




Quantum computers can solve certain problems more efficiently than classic computers. Earlier instances are Grover's quantum search~\cite{grover1996fast} and Shor's factorization~\cite{shor1999polynomial}; more recent work addresses the efficient solution of linear equations~\cite{harrow2009quantum} and the simulation of differential equations~\cite{liu2021efficient}.
Despite its potential and increasing interest from a commercial point of view, quantum computing is still in its infancy. The number of qubits of current quantum computers is prohibitively small; furthermore, low coherence times and quantum noise lead to high error rates. Further research and improvement of quantum circuits thus hinge on the availability of efficient simulation algorithms on classic computers.

Being described by a unitary complex matrix, any quantum circuit can be simulated by means of array structures and the respective matrix operations~\cite{khammassi2017qx,steiger2018projectq,murali2020software}. Unfortunately, direct array approaches are subject to the curse of dimensionality~\cite{DBLP:journals/tcad/NiemannWMTD16} because the size of the matrix is exponential in the number of qubits. This motivated the introduction of techniques that try to overcome the exponential growth, resting, for instance, on the stabilizer formalism~\cite{aaronson2004improved}, tensor networks~\cite{PhysRevLett.91.147902,villalonga2019flexible}, path sum reductions~\cite{DBLP:journals/corr/abs-1805-06908} and decision diagrams~\cite{WilleTools,DBLP:journals/tcad/NiemannWMTD16,9308161}.

Here we study bisimulation relations for quantum circuits. Bisimulation has a long tradition in computer science \cite{sangiorgi:book}. For the purpose of this paper, the most relevant branch of research on this topic regards bisimulations for quantitative models such as probabilistic bisimulation~\cite{Larsen19911,10.1007/3-540-63166-6_14}. This is closely related to ordinary lumpability for Markov chains~\cite{BuchholzOrdinaryExact}, also known as \emph{forward} bisimulation~\cite{Feret2012137}, which yields an aggregated Markov chain by means of partitioning the original state space, such that the probability of being in each macro-state/block is equal to the sum of the probabilities of each state in that block. Exact lumpability~\cite{BuchholzOrdinaryExact}, also known as \emph{backward} bisimulation~\cite{1703385}, exploits a specific linear invariant induced by a partition of the state space such that states in the same partition block have the same probability at all time points~\cite{BuchholzOrdinaryExact}. In an analogous fashion, forward and backward bisimulations have been developed for chemical reaction networks~\cite{concur15,DBLP:conf/gecco/TognazziTTV17,DBLP:journals/bioinformatics/CardelliPTTVW21}, rule-based systems~\cite{Feret2012137,Feret21042009}, and ordinary differential equations~\cite{PNAScttv,DBLP:conf/qest/CardelliTTV18}.

In all these cases, lumping can be mathematically expressed as a specific linear transformation of the original state space into a reduced space that is induced by a partition. However, in general, lumping allows for arbitrary linear transformations~\cite{1703385,DBLP:journals/scp/Boreale20}. Since this may introduce loss of information, \emph{constrained lumping} allows one to specify a subspace of interest that should be preserved in the reduction~\cite{tomlin1997effect,DBLP:journals/scp/Boreale20,ovchinnikov2021,jimenez2022}. In partition-based bisimulations, constraints can be specified as suitable user-defined initial partitions of states for which lumping is computed as their (coarsest) refinement~\cite{lics16,Derisavi2003309}. Bisimulation relations for dynamical systems~\cite{DBLP:journals/tac/PappasLS00,DBLP:journals/lmcs/Boreale19} and the notions of constrained linear lumping in~\cite{tomlin1997effect,ovchinnikov2021,jimenez2022}, instead, can be understood as linear projections (also known as ``lumping schemes'') into a lower-dimensional system that preserves an arbitrary linear constraint subspace.

The aim of this paper is to boost the simulation of quantum circuits via
forward- and backward-type bisimulations that can be constrained to subspaces.\footnote{In the following, the simulation of quantum circuits refers to their execution on a classical computer and not to the notion of one-sided bisimulation.}

Analogously to the cited literature, with \emph{forward constrained bisimulation} (FCB) the aim is to obtain a lower-dimensional circuit which (exactly) preserves the behavior of the original circuit on the subset of interest. In \emph{backward constrained bisimulation} (BCB), the reduction is valid on the constraint subspace; in this manner, the original quantum state can be fully recovered from the reduced circuit. Overall, this setting has complementary interpretation with respect to the analysis of a quantum circuit. It is known that a quantum state can only be accessed by means of a quantum measurement, mathematically expressed as a projection onto a given subspace. FCB, in general, preserves any projection onto the constraint subspace. That is, if the constraint subspace contains the measurement subspace, FCB will exactly preserve the quantum measurement, but the full quantum state cannot be recovered in general. Instead, constraining the invariant set of BCB to contain the input of the circuit ensures that the full circuit result can be recovered from the reduced circuit. We show that FCB and BCB are related by a duality property stating that a lumping scheme is an FCB if and only if its complex conjugate transpose is a BCB. Interestingly, this is analogous to the duality established between ordinary (forward) and exact (backward) lumpability for Markov chains~\cite{Derisavi2003309,DBLP:conf/popl/CardelliTTV16}, although it does not carry over to other models in general~\cite{BBLTTV21Lics,DBLP:conf/cdc/TognazziTTV18}. A relevant implication of this result is that one only needs one algorithm to compute both FCB and BCB. As a further contribution of this paper, we present such an algorithm, developed as an adaptation of the CLUE method for the constrained lumping of systems of ordinary differential equations with polynomial right-hand sides~\cite{ovchinnikov2021} of which it inherits the polynomial-time complexity in matrix size.

To show the applicability of our constrained bisimulations, we analyze several case studies for which we report both theoretical and experimental results. Specifically, we first study three classic quantum circuits for search (Grover's algorithm~\cite[Section 6.2]{nielsen00}), optimization~\cite{farhi2014quantum}, and factorization~\cite[Section 5.3.2]{nielsen00}, respectively. In Grover's algorithm, the cardinality of the search domain is exponential in the number of qubits; we prove that BCB can always reduce the circuit matrix to a $2 \times 2$ matrix while exactly preserving the output of interest. Next, we consider the quantum approximate optimization (QAOA) algorithm for solving SAT and MaxCut instances~\cite{farhi2014quantum}; in this setting, our main theoretical result is an upper bound on the size of the reduced (circuit) matrix  by the number of clauses (SAT) or edges (MaxCut). Finally, for quantum factorization we prove that the size of the reduced matrix gives the multiplicative order, that is, it solves the order finding problem to which the factorization problem can be reduced~\cite{nielsen00}.

From an experimental viewpoint, using a prototype based on a publicly available implementation of CLUE, we compare the aforementioned theoretical bounds against the actual reductions on a set of randomly generated instances. Moreover, we conducted a large-scale evaluation on common quantum algorithms collected from the MQT Bench repository~\cite{quetschlich2022mqtbench}, showing considerable reductions in all cases. Finally, we demonstrate that constrained bisimulation complements state-of-the-art methods for quantum circuit simulation based on decision diagrams~\cite{DBLP:journals/tcad/NiemannWMTD16}, as implemented in the MQT DDSIM tool~\cite{WilleTools}.

\emph{Relation to conference version.} The present work extends the conference version~\cite{qTACAS} by computing quantum bisimulations using decision diagrams. Specifically, the exponential time and space complexity of Theorem~\ref{thm:red:measure} is improved in the new Theorem~\ref{thm:dd}. The latter states that the worst-case complexity for computing the quantum bisimulation, the computational bottleneck of~\cite{qTACAS}, is at most polynomial in the size of the largest decision diagram appearing during the computation of the quantum bisimulation. In a similar vein, the Python implementation from~\cite{qTACAS} has been re-implemented in C++ and uses decision diagrams instead of vectors. Using the improved framework, the quantum benchmarks  from~\cite{qTACAS} have been revisited. Apart from reducing the simulation times from~\cite{qTACAS} by several orders of magnitude, it is demonstrated that the combination of quantum bisimulation and decision diagrams outperforms each approach in isolation.

\emph{Further related work.} Probabilistic bisimulations~\cite{10.1007/3-540-63166-6_14,DBLP:conf/concur/BacciBLM16,DBLP:journals/lmcs/BacciBLM18} have been considered for quantum extensions of process calculi, see~\cite{10.1145/1040305.1040318,FENG20071608} and references therein. Similarly to their classic counterparts, these seek to identify concurrent (quantum) processes with similar behavior.
The current work, instead, is about boosting the simulation of quantum circuits and is in line with~\cite{DBLP:conf/qest/BacciBLM13,DBLP:conf/popl/CardelliTTV16,DBLP:journals/pe/TschaikowskiT17}. Specifically, it operates directly over quantum circuits rather than processes and exploits general linear invariants in the (complex) state space. In engineering, invariant-based reductions of linear systems are known under the names of proper orthogonal decomposition~\cite{antoulas,nielsen2009quantum}, Krylov methods~\cite{antoulas}, and dynamic mode decomposition~\cite{DMD2009,Kumar_2015,Goldschmidt_2021}. Linear invariants also describe safety properties~\cite{bookBHK} in quantum model checking~\cite{10.1145/2629680,DBLP:journals/corr/abs-1807-09466}, without being used for reduction though. $\mathcal{L}$-bisimulation~\cite{DBLP:journals/scp/Boreale20} and \cite{Kumar_2015,Goldschmidt_2021} yield the same reductions, with the difference being that the former obtains the smallest reduction up to a given initial constraint, while the latter computes the smallest reduction up to an initial condition. While similarly to us relying on reduction techniques, \cite{Kumar_2015,Goldschmidt_2021} focus on the reduction of quantum Hamiltonian dynamics, with applications mostly in quantum physics and chemistry. Instead, we study the reduction of quantum circuits which are the prime citizens of quantum computing. Moreover, we provide a prototype implementation of our approach and perform a large-scale numerical evaluation.

\emph{Paper outline.} The paper is structured as follows. After a review of core concepts, Section~\ref{sec:prel} introduces forward and backward constrained bisimulation of (quantum) circuits. There, we also provide an algorithm for the computation of constrained bisimulations by extending~\cite{ovchinnikov2021,leguizamon2023} to circuits. Section~\ref{sec:applications} then derives bounds on the reduction sizes of quantum search~\cite{grover1996fast}, quantum optimization~\cite{farhi2014quantum} and quantum order finding~\cite{nielsen00}. Section~\ref{sec:benchmark}, instead, conducts a large-scale evaluation on published quantum benchmarks~\cite{quetschlich2022mqtbench} and compares constrained bisimulations with MQT DDSIM with respect to the possibility of speeding up circuit simulations. The paper concludes in Section~\ref{sec:conc}.

\section{Preliminaries}

\subsubsection*{Notation.}  We shall denote by $n$ the number of qubits and set $N = 2^n$ for convenience. Column vectors are denoted by the \emph{ket} notation $\ket{z}$, while the complex conjugate transpose of $\ket{z}$ is denoted by $\ket{z}^\dagger = \bra{z}$, i.e., $\bra{z} = \ket{\bar{z}}^T$ with $\bar{\cdot}$ and $\cdot^T$ denoting complex conjugation and transpose, respectively. In a similar vein, $\bra{z} \, \ket{z} = \braket{z}$, where $\langle \cdot \mid \cdot \rangle$ is the standard scalar product over $\CO^N$. Following standard notation, the canonical basis vectors of $\CO^N$ are expressed using tensor products and bit strings $x \in \{0,1\}^n$; specifically, writing $\otimes$ for the Kronecker product, we have $\ket{x_n} \otimes \ket{x_{n-1}} \otimes \ldots \otimes \ket{x_1} = \ket{x_{n}} \ket{x_{n-1}} \ldots \ket{x_1} = \ket{x_n x_{n-1} \ldots x_1} = \ket{d}$, where $0 \leq d \leq 2^n - 1$ is a decimal representation of $x$, see~\cite{nielsen00} for details. We usually denote by $\ket{x}$ canonical basis vectors with $x \in \{0,1\}^n$, whereas $\ket{u}, \ket{v}, \ket{w}, \ket{z} \in \CO^N$ refer to linear combinations in the form  $\ket{z} = \sum_{x \in \{0,1\}^n} c_x \ket{x}$ with $c_x \in \CO$. For any canonical basis vector, we have $\ket{x} = \ket{\bar{x}}$. To avoid confusion, forward constrained bisimulations (FCB) are denoted by row matrices $L \in \CO^{d \times N}$ with $d \leq N$, while backward constrained bisimulations (BCB) are denoted by column matrices $L^\dagger \in \CO^{N \times d}$.

\subsection{Linear Algebra and Dynamical Systems} 

We begin by introducing core concepts from linear algebra and quantum computing~\cite{meyer01,nielsen00}.

\begin{definition}[Core Concepts]
\begin{itemize}
    \item The column space of a matrix $M$ are all linear combinations of its columns and is denoted by $\langle M \rangle_c$. One says, the columns of $M$ span $\langle M \rangle_c$.
    \item The row space of a matrix is the set of all linear combinations of its rows and is denoted by $\langle M \rangle_r$. One says, the rows of $M$ span $\langle M \rangle_r$.
    \item A (quantum) circuit over $n$ qubits is described by a unitary map $U \in \CO^{N \times N}$,  that is, $U^{-1} = U^\dagger$.
    \item A (quantum) state $\ket{z} \in \CO^N$ is a vector with (Euclidian) norm one.
    \item A matrix $P \in \CO^{N \times N}$ is an orthogonal projection if $P \circ P = P = P^\dagger$.
    \item Any vector $\ket{z} \in \CO^N$ generates the linear subspace $S_{\ket{z}} = \langle \ket{z} \rangle_c$.
\end{itemize}
\end{definition}

Throughout the paper, we do not work at the higher level where quantum circuits are defined by means of quantum gate compositions~\cite{nielsen00}. Instead, we work directly at the level of the unitary maps that are induced by such compositions. With this in mind, we use the terms ``unitary map'' and ``quantum circuit'' interchangeably.

We distinguish between one- and multi-step applications of a quantum circuit~\cite{nielsen00}. For an input state $\ket{w_0} \in \CO^N$, the full quantum state after \emph{one-step} application is $U \ket{w_0}$. Instead, the full quantum state after a \emph{multi-step} application is given by $U^k \ket{w_0}$, where $k > 1$ is the number of steps. These definitions justify interpreting a quantum circuit as a discrete-time dynamical system as follows.

\begin{definition}[Dynamical System]
A circuit $U \in \CO^{N \times N}$ with input state $\ket{w_0}$ induces the discrete time dynamical system (DS) $\ket{w_{k+1}} = U \ket{w_k}$, with $k \geq 0$. We call $\ket{w_k}$ the \emph{full quantum state} at step $k$.
\end{definition}

\begin{example}\label{ex:pauli:x}
The one-qubit circuit
$U = \bigl( \begin{smallmatrix} 0 & 1 \\ 1 & 0 \end{smallmatrix}\bigr)$ is known as the Pauli $X$-gate~\cite{nielsen00}.
In the case of $k \geq 1$ steps and input $\ket{w_0} = \ket{\phi}$, where $\ket{\phi} = (1,-1) / \sqrt{2}$, the induced DS can be shown to be $\ket{w_k} = (-1)^k \ket{\phi}$.
\end{example}

The result of a quantum computation is not directly accessible and is usually queried using quantum measurements~\cite{nielsen00}. These can be described by projective measurements~\cite{nielsen00}, formally given by a family of orthogonal projections $\{P_1, \ldots, P_m\}$ that satisfy $P_1 + \ldots + P_m = I$.   When a quantum state $\ket{z} \in \CO^N$ is measured, the probability of outcome $1 \leq i \leq m$ is $\pi_i = \bra{z} P_i \ket{z}$. In case of outcome $i$, the quantum state after measurement is $P_i \ket{z} / \sqrt{\pi_i}$. We will be primarily concerned with the case $\{P,I-P\}$ for a given orthogonal projection $P$.

Often, one is interested in querying states from a specific subspace $S$. For example, the result of the HHL algorithm~\cite{harrow2009quantum}, considered in Section~\ref{sec:benchmark}, is stored in a subset of all qubits, i.e., in a subspace. To this end, it suffices to use a projective measurement that identifies $S$.

\begin{definition}
Given an orthogonal projection $P$, we call $P \ket{z}$ the $P$-measure\-ment of $\ket{z}$. A subspace $S \subseteq \CO^N$ is identifiable by $P$ if $P \ket{z} = \ket{z}$ for all $\ket{z} \in S$.
\end{definition}

A particularly simple but useful class of projective measurements is the one that measures a single state $\ket{w}$, i.e., that identifies the space $S_{\ket{w}}$ spanned by $\ket{w}$. This is given by the orthogonal projection $P_{\ket{w}} := \ket{w} \bra{w}$.

\begin{example}\label{ex:pauli:x:2}
Assume that we are interested in measuring the result of Example~\ref{ex:pauli:x} using measurement $P_{\ket{\phi}}$ that identifies $S_{\ket{\phi}}$. Then, for $\ket{w_0} = \binom{1}{0}$, it holds that $P_{\ket{\phi}} \ket{w_k} = (-1)^k \ket{\phi} / \sqrt{2}$.
\end{example}

\subsection{Linear Algebra with Decision Diagrams} 

In the following we briefly review decision diagrams and outline how these can be used to conduct linear algebra operations efficiently; see~\cite{wille2023hamiltonian} for a detailed discussion. We start by noting that a qubit $\ket{z} = \alpha_0 \ket{0} + \alpha_1 \ket{1}$ can be conveniently represented as a decision diagram
\begin{equation*}
	\begin{tikzpicture}[node distance=0.5 and 0.5]
		\node (s0) {$\ket{z} = \alpha_0 \ket{0} + \alpha_1 \ket{1} = \begin{pmatrix}
	\alpha_0 & \alpha_1
\end{pmatrix}^\top\equiv$};
		\node[vertex,right=of s0]  (dd0) {};
		\node[terminal, below=of dd0] (t0) {};
		\draw[edge={1}{0}] ($(dd0)+(0,0.5cm)$) -- (dd0);
		\draw[edge0={0.707}{0}] (dd0) to node[midway, left] {$\alpha_0$} (t0);
		\draw[edge1={0.707}{0}] (dd0) to node[midway, right] {$\alpha_1$} (t0);
	\end{tikzpicture}.
\end{equation*}
It consists of a single node with one incoming edge that represents the entry point into the decision diagram as well as two successors that represent the split between the basis states of the single qubit, ending in a terminal node denoted by the black box. The state's amplitudes are annotated to the respective edges. Edges without annotations correspond to an edge weight of 1. 

Building on the intuition from a single-qubit state, we can move to larger systems.

		


\begin{example}\label{ex:multiqubitdd}
	Consider the following state vector of a three-qubit system:
	\begin{equation}
		\ket{z} = \begin{pmatrix} \frac{1}{2\sqrt{2}} & \frac{1}{2\sqrt{2}} &  \frac{1}{2} & 0 & \frac{1}{2\sqrt{2}} & \frac{1}{2\sqrt{2}} & \frac{1}{2} & 0\end{pmatrix}^T
	\end{equation}
	Then, $\ket{z}$ can be recursively split into equally-sized parts, essentially making a case distinction over the qubit value $q_i \in \{ 0, 1 \}$, for each $i =2,1,0$:
	\begin{equation*}
                \overbrace{
                \overbrace{\begin{matrix}
                \overbrace{\begin{matrix}
                \bigl( \overset{\ket{000}}{
                \begin{matrix} \frac{1}{2\sqrt{2}}
                \end{matrix} }
                & \overset{\ket{001}}{
                \begin{matrix} \frac{1}{2\sqrt{2}}
                \end{matrix} }
                \end{matrix}}^{\ket{00q_0}}
                & \overbrace{
                \begin{matrix}\overset{\ket{010}}{
                \begin{matrix} \frac{1}{2}
                \end{matrix}}
                & \overset{\ket{011}}{
                \begin{matrix} 0
                \end{matrix} }
                \end{matrix}}^{\ket{01q_0}}
                \end{matrix}}^{\ket{0q_1q_0}}
                \ \
                \overbrace{\begin{matrix}
                \overbrace{\begin{matrix}
                \overset{\ket{100}}{
                \begin{matrix} \frac{1}{2\sqrt{2}}
                \end{matrix} }
                & \overset{\ket{101}}{
                \begin{matrix} \frac{1}{2\sqrt{2}}
                \end{matrix} }
                \end{matrix}}^{\ket{10q_0}}
                & \overbrace{
                \begin{matrix} \overset{\ket{110}}{
                \begin{matrix} \frac{1}{2}
                \end{matrix} }
                & \overset{\ket{111}}{
                \begin{matrix} 0
                \end{matrix} } \bigr)^\top
                \end{matrix}}^{\ket{11q_0}}
                \end{matrix}}^{\ket{1q_1q_0}}
                }^{\ket{q_2q_1q_0}} \\\\\\\\\ ,
    \end{equation*}
    This directly translates to the decision diagram:
 	
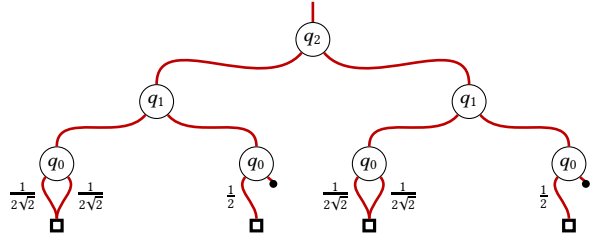
\begin{equation*}
 		\begin{tikzpicture}[node distance=0.5 and 1.0]
 		\node[vertex] (q2) {$q_2$};
 		\node[vertex,below left=0.5 and 1.75 of q2] (q1a) {$q_1$};
 		\node[vertex,below right=0.5 and 1.75 of q2] (q1b) {$q_1$};
 		\node[vertex,below left= of q1a] (q0a) {$q_0$};
 		\node[vertex,below right= of q1a] (q0b) {$q_0$};
 		\node[vertex,below left= of q1b] (q0c) {$q_0$};
 		\node[vertex,below right= of q1b] (q0d) {$q_0$};
 		\node[terminal, below= of q0a] (ta) {};
 		\node[terminal, below= of q0b] (tb) {};
 		\node[terminal, below= of q0c] (tc) {};
 		\node[terminal, below= of q0d] (td) {};
 	
		\draw[edge={1}{0}] ($(q2)+(0,0.5cm)$) -- (q2);
		
		\draw[edge0={1}{0}] (q2) to (q1a);		      
            \draw[edge1={1}{0}] (q2) to (q1b);
		
		\draw[edge0={1}{0}] (q1a) to (q0a);
		\draw[edge1={1}{0}] (q1a) to (q0b);
		
		\draw[edge0={1}{0}] (q1b) to (q0c);			
		\draw[edge1={1}{0}] (q1b) to (q0d);
		
		\draw[edge0={1}{0}] (q0a) to node[midway, left] {$\frac{1}{2\sqrt{2}}$} (ta);
  		\draw[edge1={1}{0}] (q0a) to node[midway, right] {$\frac{1}{2\sqrt{2}}$} (ta);

  		\draw[edge0={1}{0}] (q0b) to node[midway,left] {$\frac{1}{2}$} (tb);
      	\draw[edge1={1}{0}] (q0b) to ++(-50:0.35)  node[zerostub] {};

        \draw[edge0={1}{0}] (q0c) to node[midway, left] {$\frac{1}{2\sqrt{2}}$} (tc);
      	\draw[edge1={1}{0}] (q0c) to node[midway, right] {$\frac{1}{2\sqrt{2}}$} (tc);
       
    	\draw[edge0={1}{0}] (q0d) to node[midway,left] {$\frac{1}{2}$} (td);
        \draw[edge1={1}{0}] (q0d) to ++(-50:0.35) node[zerostub] {};
	\end{tikzpicture}
 	\end{equation*}
 	Each level of the decision diagram consists of decision nodes with corresponding left and right successor edges. These successors represent the path that leads to an amplitude where the local quantum system (corresponding to the \emph{level} of the node, annotated here with the labels) is in the $\ket{0}$ (left successor) or the $\ket{1}$ state (right successor).
\end{example}
 	
At this point, this has been just a one-to-one translation between the state vector and a graphical representation.
The unique feature of decision diagrams is that their graph structure allows redundant parts to be merged in the representation instead of being repeated.

\begin{example}\label{ex:reduction}
	Observe how, in the previous example, the left and right successors of the top-level ($q_0$) node lead to exactly the same structure.
	As a result, the whole sub-diagram does not need to be represented twice, i.e.,
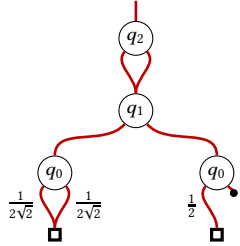
	\begin{equation*}
		\begin{tikzpicture}[node distance=0.5 and 0.75]
 		\node[vertex] (q2) {$q_2$};
 		\node[vertex,below= of q2] (q1) {$q_1$};
 		\node[vertex,below left= of q1] (q0a) {$q_0$};
 		\node[vertex,below right= of q1] (q0b) {$q_0$};
 		\node[terminal, below= of q0a] (ta) {};
 		\node[terminal, below= of q0b] (tb) {};
 	
		\draw[edge={1}{0}] ($(q2)+(0,0.5cm)$) -- (q2);
		
		\draw[edge0={1}{0}] (q2) to (q1);
		\draw[edge1={1}{0}] (q2) to (q1);
		
		\draw[edge0={1}{0}] (q1) to (q0a);
		\draw[edge1={1}{0}] (q1) to (q0b);
		
		\draw[edge0={1}{0}] (q0a) to node[midway, left] {$\frac{1}{2\sqrt{2}}$} (ta);
  		\draw[edge1={1}{0}] (q0a) to node[midway, right] {$\frac{1}{2\sqrt{2}}$} (ta);

  		\draw[edge0={1}{0}] (q0b) to node[midway,left] {$\frac{1}{2}$} (tb);
      	\draw[edge1={1}{0}] (q0b) to ++(-50:0.35)  node[zerostub] {};
	\end{tikzpicture}
	\end{equation*}
	From a memory perspective, this reduction  compressed the memory required to represent the state by 50\%.
\end{example}

Intuitively, the idea is to represent state vectors compactly as decision diagrams by halving the vector in a recursive fashion until all state vector redundancies have been exploited. Although substantial compression is often possible, we mention that there exist state vectors yielding decision diagrams of exponential size in the number of qubits~\cite{DBLP:journals/tcad/ZulehnerW19}.

Merely defining means for compactly representing states is not sufficient, and it is crucial also to define efficient means to work with or manipulate the resulting representations. In~\cite{wille2023hamiltonian} it has been outlined how decision diagrams can be used to perform linear algebra operations efficiently. There, it has been argued that decision diagrams support addition, scalar multiplication, and the scalar product. The main idea behind the efficient implementation is to recursively break the respective operations down into sub-computations. 

We shall sketch the addition of decision diagrams next. The idea is to recursively rewrite it by noting that 
\begin{equation*}
        \ket{z} + \ket{z'} = \begin{pmatrix} z_0 \\ z_1 \end{pmatrix} + \begin{pmatrix} z'_0 \\ z'_1 \end{pmatrix} 
         = w \begin{pmatrix} \alpha_0  \\ \alpha_1 \end{pmatrix} + w' \begin{pmatrix} \alpha'_0 \\ \alpha'_1 \end{pmatrix} 
         = \begin{pmatrix} w \alpha_0 + w' \alpha'_0 \\ w \alpha_1 + w' \alpha'_1 \end{pmatrix},
\end{equation*}
where $w$ and $w'$ are common factors of the terms in $\ket{z}$ and $\ket{z'}$, respectively. In the decision-diagram formalism, this corresponds to a simultaneous traversal of both decision diagrams from their roots to the terminal (multiplying edge weights along the way until the individual amplitudes are reached) and back again (accumulating the results of the recursive computations).
More precisely,
\begin{equation}
	\begin{tikzpicture}[node distance=0.5 and 0.25]
    	\node[vertex] (top0) {};
    	\node[vertex, dashed, below left=of top0] (t0) {$z_0$};
    	\node[vertex, dashed, below right=of top0] (t1) {$z_1$};
    	\draw[edge={1}{0}] ($(top0)+(0,0.5cm)$) to node[midway,right] {$w$} (top0);
    	\draw[edge0={1}{0}] (top0) to node[midway,left] {$\alpha_0$} (t0);
    	\draw[edge1={1}{0}] (top0) to node[midway,right] {$\alpha_1$} (t1);
    	
    	\node[vertex,right=2 of top0] (top1) {};
    	\node[vertex, dashed, below left=of top1] (t2) {$z'_0$};
    	\node[vertex, dashed, below right=of top1] (t3) {$z'_1$};
    	\draw[edge={1}{0}] ($(top1)+(0,0.5cm)$) to node[midway,right] {$w'$} (top1);
    	\draw[edge0={1}{0}] (top1) to node[midway,left] {$\alpha'_0$} (t2);
    	\draw[edge1={1}{0}] (top1) to node[midway,right] {$\alpha'_1$} (t3);
    	
    	\node[] at ($(top0)!0.5!(top1)$) (plus) {$+$};
    	
    	\node[vertex,below=1.5 of plus] (top2) {};
    	\node[rectangle, draw, dashed, below left=of top2] (t4) {
    	\begin{tikzpicture}
    	\node[vertex, dashed] (a) {$z_0$};
    	\draw[edge={1}{0},solid] ($(a)+(0,0.5cm)$) to node[midway,right] {$w\alpha_0$} (a);
    	\node[vertex, dashed,right=1.0 of a] (b) {$z'_0$};
    	\draw[edge={1}{0},solid] ($(b)+(0,0.5cm)$) to node[midway,right] {$w'\alpha'_0$} (b);
    	\node[right=0.5 of a] {$+$};
		\end{tikzpicture}};
    	\node[rectangle, draw, dashed, below right=of top2] (t5) {
    	\begin{tikzpicture}
    	\node[vertex, dashed] (a) {$z_1$};
    	\draw[edge={1}{0},solid] ($(a)+(0,0.5cm)$) to node[midway,right] {$w\alpha_1$} (a);
    	\node[vertex, dashed,right=1.0 of a] (b) {$z'_1$};
    	\draw[edge={1}{0},solid] ($(b)+(0,0.5cm)$) to node[midway,right] {$w'\alpha'_1$} (b);
    	\node[right=0.5 of a] {$+$};
		\end{tikzpicture}};
    	\draw[edge={1}{0}] ($(top2)+(0,0.5cm)$) -- (top2);
    	\draw[edge0={1}{0}] (top2) to (t4);
    	\draw[edge1={1}{0}] (top2) to (t5); 
    	
    	\node[] at ($(plus)!0.6!(top2)$) {$=$};
		\end{tikzpicture} \\\\\ ,
\end{equation}
where the dashed nodes represent the respective successor decision diagrams.
Overall, this results in a complexity that is linear in the size of the larger decision diagram. Scalar multiplication and scalar product of decision diagrams can be treated in a similar fashion and yield the same complexity~\cite{wille2023hamiltonian}.

\section{Constrained Bisimulations for Quantum Circuits}\label{sec:prel}

\subsection{Forward Constrained Bisimulation}

We next introduce FCB. 

\begin{definition}[Forward Constrained Bisimulation, FCB]\label{def:cons:lump}
Fix a circuit defined by $U \in \CO^{N \times N}$ with initial state $\ket{w_0}$ and a matrix $L \in \CO^{d \times N}$ with orthonormal rows.
\begin{enumerate}[a)]
    \item The DS reduced by $L$ is given by $\ket{\hat{w}_{k+1}} = \hU \ket{\hat{w}_k}$, where $\hU = L U L^\dagger$, and initial state $\ket{\hat{w}_0} = L\ket{w_0}$.
    \item $L$ is called forward constrained bisimulation of DS $\ket{w_{k+1}} = U \ket{w_k}$ wrt constraint subspace $S \subseteq \CO^N$ when $S \subseteq \langle L^\dagger \rangle_c$ and $L \ket{w_k} = \ket{\hat{w}_k}$ for all $k \geq 1$.
\end{enumerate}
\end{definition}

Before commenting on the definition, we establish the following.

\begin{lemma}\label{lem:unit}
The reduced map $\hat{U}$ in Definition~\ref{def:cons:lump} is unitary.
\end{lemma}
\begin{proof}
See proof of Theorem~\ref{thm:duality}.  
\end{proof}

We note that the reduction holds for any choice of initial state $\ket{w_0}$, analogously to the aforementioned forward-type bisimulations~\cite{PNAScttv,DBLP:conf/popl/CardelliTTV16} for (real-valued) dynamical systems. The assumption of orthonormality of rows of $L$ implies that $d \leq N$, i.e., $L$ is a transformation onto a possibly smaller-dimensional state space. Although it can be dropped without loss of generality~\cite{ovchinnikov2021}, it allows a more immediate relation to projective measurements. In fact, matrix $L$ induces the orthogonal projection $P_L$ defined as $P_L = L^\dagger L$. This projective measurement identifies $S$ because $S \subseteq \langle L^\dagger \rangle_c$. Moreover, $P_L \ket{w_k}$ is preserved in the reduced system for any $k$. To see this, it suffices to multiply $L \ket{w_k} = \ket{\hat{w}_k}$ by $L^\dagger$ from the left and to note that this yields $P_L \ket{w_k} = L^\dagger \ket{\hat{w}_k}$. 

\begin{example}\label{ex:lump}
Continuing Example~\ref{ex:pauli:x}, it can be shown that the $2 \times 1$ matrix $L = \ket{\phi}^\dagger = (1,-1) / \sqrt{2}$ is an FCB wrt $S_{\ket{\phi}}$. In fact, since $\hU = L U L^\dagger = -1$, we obtain $\ket{\hat{w}_{k+1}} = - \ket{\hat{w}_k}$, while a direct calculation confirms that $L \ket{w_0} = \ket{\hat{w}_0}$ implies $L \ket{w_{k}} = \ket{\hat{w}_{k}}$ for all $k > 0$.
Multiplying both sides by $L^\dagger$ from the left yields $L^\dagger L U^k \ket{w_0} = (-1)^k L^\dagger L \ket{w_0}$. Consequently, the $P_{\ket{\phi}}$-measurement of the original map can be obtained from the $P_{\ket{\phi}}$-measurement of the reduced map. \end{example}

Algorithm~\ref{alg} adapts the algorithm for (real-valued) systems of ordinary differential equations with polynomial derivatives developed in~\cite{ovchinnikov2021,leguizamon2023} to the complex domain and yields the minimal FCB wrt subspace $S$, i.e., it returns an orthonormal $L \in \CO^{d \times N}$ whose dimension $d$ is minimal. 

\begin{algorithm}[t]
		\caption{\footnotesize Computation of an FCB $L$ wrt subspace $S$}\label{alg}
		

	\begin{algorithmic}[1]

		\REQUIRE Unitary map $U \in \CO^{N \times N}$ and subspace $S \subseteq \CO^N$.
  
		\STATE \textbf{compute} orthonormal basis of $S$, store it in column matrix $L^\dagger \in \CO^{N \times d_0}$ \label{alg:init} 

		\REPEAT\label{alg:main:loop} 
  
		\FORALL{columns $\ket{z}$ of $L^\dagger$}\label{alg:loop} 

            \STATE {\textbf{compute} $\ket{\pi} = P_L U \ket{z}$}
  \label{alg:acl:proj}
		\IF{$\ket{\pi} \neq U \ket{z}$}\label{alg:acl:check}
            \STATE $\ket{w} = U \ket{z} - \ket{\pi} $ \label{alg:orth:comp}
		\STATE {\textbf{append} column $\ket{w} / \bra{w} \ket{w}$ to $L^\dagger$}\label{alg:append}
		\ENDIF 

		\ENDFOR

		\UNTIL{no columns have been appended to $L^\dagger$}
  
        \RETURN matrix $L^{\dagger \dagger}$.

	\end{algorithmic}

 \end{algorithm}

\begin{theorem}[Minimal FCB]\label{thm:red:measure}
Algorithm~\ref{alg} computes a minimal FCB $L \in \CO^{d \times N}$ w.r.t. subspace $S$, that is, the rowspace of any FCB $L'$ w.r.t. $S$ contains that of $L$. The complexity of Algorithm~\ref{alg} is polynomial in $N$. 
\end{theorem}
\begin{proof}
See proof of Theorem~\ref{thm:duality}.
\end{proof}

We briefly comment on Algorithm~\ref{alg}. The idea behind it exploits the fact that $L$ can be shown to be an FCB whenever $L^\dagger$ is an invariant set of the map $U$, that is, if the column space of $L^\dagger U$ is contained in the column space of $L^\dagger$.

The algorithm begins by initializing $L^\dagger$ with a basis of $S$ in line~\ref{alg:init}. This ensures that $S$ is contained in the column space of the final result. For every column $\ket{z}$ of $L^\dagger$, the main loop in line~\ref{alg:main:loop} checks whether $U \ket{z}$ is in the column space of $L^\dagger$ (line~\ref{alg:acl:check}) by computing its projection $\ket{\pi}$ onto the column space of $L^\dagger$ (line~\ref{alg:acl:proj}). If it is not in the column space, the projection will differ from $U \ket{z}$ and the residual $\ket{w}$ must be added to $L^\dagger$. This shows correctness, while the minimality of FCB $L$ follows from the fact that only the necessary residuals are added to $L^\dagger$. The complexity of the algorithm, instead, follows by noting that at most $N$ columns can be added to $L^\dagger$ and that all computations of the main loop require, similarly to the computation in line~\ref{alg:init}, at most $\mathcal{O}(N^3)$ operations.

\begin{remark}\label{rem:single}
As can be noticed in Algorithm~\ref{alg}, e.g., line 4, the computation of an FCB subsumes the computation of a single step of the circuit. For practical applications to single-step circuits where the modeler is interested in only a single input, FCB may be as expensive as simulating the original circuit directly. Hence, it is obvious that the effectiveness of constrained bisimulations is particularly relevant when simulating the circuit with respect to several inputs, or when considering multi-step applications. Examples of this are provided in Section~\ref{sec:applications} and a numerical evaluation is carried out in Section~\ref{sec:benchmark}.
\end{remark}

\begin{example}\label{ex:lump:compt}
Consider the FCB $L = \ket{\phi}^\dagger$ w.r.t. subspace $S_{\ket{\phi}}$ from Example~\ref{ex:lump}. Then, noting that $(I - P_L) \ket{\phi} = 0$, we infer that Algorithm~\ref{alg} terminates in line~\ref{alg:acl:check}. Hence, $L$ is a minimal FCB w.r.t. $S_{\ket{\phi}}$.
\end{example}
	
\subsection{Backward Constrained Bisimulation} BCB yields a reduced system through the identification of an invariant set, i.e., a subspace $S$ such that $U^k \ket{z} \in S$ for any $\ket{z} \in S$ and $k \geq 1$. Whereas in FCB the reduced model can recover projective measurements onto the constraint set \emph{for any} initial set, here one can recover the full quantum state, so long as the initial states belong to the invariant set.

\begin{definition}[Backward Constrained Bisimulation, BCB]\label{def:dmd}
Let $U, L$ and $\hU$ be as in Definition~\ref{def:cons:lump}. Then, $L^\dagger$ is a BCB of the dynamical system $\ket{w_{k+1}} = U \ket{w_k}$ w.r.t. a subspace of inputs $S \subseteq \CO^N$ when $S \subseteq \langle L^\dagger \rangle_c$ and whenever $\ket{w_0} = L^\dagger \ket{\hat{w}_0}$ implies $\ket{w_k} = L^\dagger \ket{\hat{w}_k}$ for all $k \geq 1$.
\end{definition}

Similarly to FCB, we assume without loss of generality that $L \in \CO^{d \times N}$ has orthonormal rows. 
As anticipated above, FCB and BCB are not comparable in general. In fact, an FCB $L$ does not make any assumptions on the initial condition $\ket{w_0}$, while a BCB $L^\dagger$ does so by requiring $L^\dagger L \ket{w_0} = \ket{w_0}$. Conversely, a BCB $L^\dagger$ allows one to obtain $\ket{w_k}$, while an FCB $L$ allows one to obtain $L \ket{w_k}$ instead of $\ket{w_k}$ itself. 

\begin{example}\label{ex:dmd}
Fix $\ket{\phi} = (1,-1)^T / \sqrt{2}$ from Example~\ref{ex:lump:compt} and recall that $L = \ket{\phi}^\dagger$, $U \ket{\phi} = - \ket{\phi}$ and $\hU = -1$. Then, $L^\dagger$ is a BCB of $U$ w.r.t. $S_{\ket{\phi}}$. In fact, $L^\dagger L \ket{w_0} = \ket{w_0}$ implies $\ket{w_0} = \ket{\phi}$, while
\[
L^\dagger \ket{\hat{w}_k} = (-1)^k L^\dagger \ket{\hat{w}_0} = (-1)^k L^\dagger L \ket{w_0} = (-1)^k \ket{w_0} = U^k \ket{\phi} = \ket{w_k} .
\]
\end{example}

Example~\ref{ex:dmd} anticipates the next result that states FCB and BCB are dual notions. This generalizes the known duality of ordinary and exact lumpability of Markov chains~\cite{Derisavi2003309,DBLP:conf/popl/CardelliTTV16}. 

\begin{theorem}[Duality]\label{thm:duality}
Fix a unitary map $U \in \CO^{N \times N}$ and a subspace $S \subseteq \CO^N$. $L$ is an FCB wrt $S$ if and only if $L^\dagger$ is a BCB wrt $S$.
\end{theorem}
\begin{proof}
Let $S_0 \subseteq S$ be a basis of some fixed $S \subseteq \CO^N$. We first note that the discussion of~\cite{ovchinnikov2021,jimenez2022,leguizamon2023} and~\cite{DMD2009,antoulas} can be directly extended to the complex field. With this, we obtain:
\begin{enumerate}
    \item $L \in \CO^{d \times N}$ is an FCB wrt $S$ if and only if $\langle L U \rangle_r \subseteq \langle L \rangle_r$ with $\langle S^\dagger_0 \rangle_r \subseteq \langle L \rangle_r$.
    \item $D \in \CO^{N \times d}$ is a BCB wrt $S$ if and only if $\langle U D \rangle_c \subseteq \langle D \rangle_c$ with $\langle S_0 \rangle_c \subseteq \langle D \rangle_c$.
\end{enumerate}
Moreover, we observe the following:
\begin{align*}
\langle L U \rangle_r & \subseteq \langle L \rangle_r \Leftrightarrow [U \text{ bijection}] \\
\langle L U \rangle_r & = \langle L \rangle_r \Leftrightarrow [\text{daggering}] \\
\langle U^\dagger L^\dagger \rangle_c & = \langle L^\dagger \rangle_c \Leftrightarrow [\text{$U$ unitary}] \\
\langle U^{-1} L^\dagger \rangle_c & = \langle L^\dagger \rangle_c \Leftrightarrow [U \text{ bijection}] \\
\langle L^\dagger \rangle_c & = \langle U L^\dagger \rangle_c \Leftrightarrow [U \text{ bijection}] \\
\langle U L^\dagger \rangle_c & \subseteq \langle L^\dagger \rangle_c
\end{align*}
This yields Theorem~\ref{thm:duality}, i.e., $L \in \CO^{d \times N}$ is an FCB of $U$ wrt constraint $S$ if and only if $L^\dagger \in \CO^{N \times d}$ is a BCB of $U$ wrt $S$ (because $S_0^{\dagger \dagger} = S_0$). Moreover, if $L^\dagger$ is computed by Algorithm~\ref{alg}, then $L^\dagger$ is a BCB wrt $S$, while $L^{\dagger \dagger}$ is an FCB wrt $S$. This follows by noting that in such a case $L^\dagger \in \CO^{N \times d}$ satisfies
\begin{align*}
\langle L^\dagger \rangle_c & = \langle U^k \ket{z} \mid 0 \leq k \leq N-1, \ket{z} \in S \rangle_c \\
& = \langle U^k \ket{z} \mid 0 \leq k \leq N-1, \ket{z} \in S_0 \rangle_c
\end{align*}
The complexity follows from the discussion after Theorem~\ref{thm:red:measure}. A detailed complexity discussion can be obtained in~\cite{ovchinnikov2021}. 
Exploiting that an FCB $L$ satisfies $L U L^\dagger L = L U$ by~\cite{tomlin1997effect}, we obtain
\[
(L U L^\dagger)^\dagger (L U L^\dagger) = (L U^\dagger L^\dagger) (L U L^\dagger) = L U^\dagger U L^\dagger = L L^\dagger = I_{d \times d} ,
\]
showing that $\hU$ is unitary.
\end{proof}

In light of the above result, we often speak of a (constrained bisimulation) reduction. Moreover, we note that Theorem~\ref{thm:duality} ensures that a BCB reduction up to input yields an FCB reduction up result, which is discussed next.

\begin{remark}
Let $L^\dagger$ be the BCB of $U$ w.r.t. $S_{\ket{w_0}}$, where $\ket{w_0}$ is the input. Then, $L = L^{\dagger \dagger}$ is an FCB wrt $S_{\ket{w_0}}$, implying that $P_L = L^\dagger L$ identifies the column space of $L^\dagger$, see discussion after Definition~\ref{def:cons:lump}. At the same time, the result $U^k \ket{w_0}$ is in the column span of $L^\dagger$ because $U^k \ket{w_0} = \ket{w_k} = L^\dagger \ket{\hat{w}_k} = L^\dagger \hU L \ket{w_0}$.
\end{remark}

We end the section by pointing out that, thanks to Theorem~\ref{thm:duality}, Algorithm~\ref{alg} can be used to compute a minimal BCB $L^\dagger$ w.r.t. subspace $S$. In fact, the only difference is that one should return $L^\dagger$ instead of $L^{\dagger \dagger}$ in the last line of the algorithm. With this change, we notice that Algorithm~\ref{alg} coincides, in the case of a one-dimensional subspace $S \subseteq \CO^N$, with the Krylov subspace~\cite{antoulas} that can be obtained by the Arnoldi iteration~\cite{DMD2009}.

\subsection{Lumpings with Decision Diagrams}

The next result ensures that decision diagrams~\cite{DBLP:journals/tcad/NiemannWMTD16} can be used to boost the computation of quantum lumpings. 

\begin{theorem}[Lumpings with Decision Diagrams]\label{thm:dd} 
Fix a quantum circuit $C$ 
that induces the unitary map $U \in \CO^{N \times N}$ and let $S_{\ket{z}} \subseteq \CO^N$ be the subspace spanned by $\ket{z} \in \CO^N$. Assume that the quantum lumping $L$ of $U$ w.r.t. $S_{\ket{z}}$ has dimension $d$ and that $U^1 \ket{z}, \ldots, U^d \ket{z}$ can be computed and stored as decision diagrams in time and space $\calO(s)$ using circuit $C$.
\begin{enumerate}
    \item Without assuming knowledge of $d$, the quantum lumping $L$ and the reduced unitary map $\hU$ can be computed in time $\calO(s d^2)$ and space $\calO(s d + d^2)$.
    \item Provided that a vector $\ket{v} \in \CO^N$ is given in terms of a decision diagram of size $\calO(s)$, vector $L^\dagger \hat{U}^k L \ket{v}$ can be computed, for any $k \geq 1$, in $\calO(sd + k d^2)$ steps and in $\calO(sd)$ space.
\end{enumerate}
\end{theorem}
\begin{proof}
Assuming that two vectors of $\CO^N$ are represented by decision diagrams of size $\calO(s)$, it can be shown~\cite{wille2023hamiltonian} that their sum, scalar product and scalar multiplication can be computed in time $\calO(s)$ and stored as decision diagrams of size $\calO(s)$. Next, we prove the special case $\kappa = 1$.

To show 1., we set $u_i := U^i \ket{z}$ for $i \geq 0$ and note that, by assumption, $u_d$ is a linear combination of $u_0,\ldots,u_{d-1}$. Hence, if applied to $u_0,\ldots,u_d$, the (numerically stable) modified Gram-Schmidt (MGS) method computes a sequence $v_0,\ldots,v_{d-1},v_d \in \CO^N$ with $v_d = 0$ and $v_i \neq 0$ for $0 \leq i < d$, allowing one thus to determine $d$ during execution. Moreover, we can set $L = (v_0,\ldots,v_{d-1})^\dagger$ because it has orthonormal rows and the same row space as $(u_0,\ldots,u_{d-1})^\dagger$. With this, we next assume that $u_0,\ldots,u_{d-1}$ are available as decision diagrams of size $\calO(s)$ and argue that MGS can be implemented over decision diagrams such that it runs in time $\calO(s d^2)$ and stores $v_0,\ldots,v_{d-1}$ as decision diagrams of size $\calO(s)$. To see this, recall that $v_k$ is arises from $v_0,\ldots,v_{k-1}$ via $v_k^{(k-1)} / \langle v_k^{(k-1)}, v_k^{(k-1)} \rangle$, where
\begin{align*}
v_k^{(1)} & = u_k - \langle u_k , v_1 \rangle v_1 \\
v_k^{(2)} & = v_k^{(1)} - \langle v_k^{(1)} , v_2 \rangle v_2 \\
& \vdots \\
v_k^{(k-1)} & = v_k^{(k-2)} - \langle v_k^{(k-2)} , v_{k-1} \rangle v_{k-1}
\end{align*}
Hence, assuming (as inductive hypothesis) that $v_0,\ldots,v_{k-1}$ can be stored as decision diagrams of size $\calO(s)$, it follows that $v_k^{(1)}, v_k^{(2)}, \ldots, v_k^{(k-1)}$ can be also stored as decision diagrams of size $\calO(s)$. This, in turn, implies that $v_k$ can be computed in $\calO(s d)$ time and space. Since this has to be done for all $k = 0,\ldots,d-1$, we obtain a time complexity of $\calO(s d^2)$ and a space complexity of $\calO(s d + d^2)$ because we need to store $\hU$ and the computation of each $v_k$ reuses the previously computed $v_0,\ldots,v_{k-1}$. Once $v_0,\ldots,v_{d-1}$ have been obtained, the reduced matrix $\hU$ can be computed in time $\calO(s d^2)$ and space $\calO(sd + d^2)$ because $\hU_{i,j} = \langle u_j , v_i \rangle$.

We next turn to claim 2. To this end, we first note that the vector $L \ket{v} \in \CO^d$ can be obtained by computing $d$ scalar products between decision diagrams of size $\calO(s)$, thus giving rise to a time and space complexity of $\calO(s d)$. Once the vector $L \ket{v}$ is known, computing $\hU^k L \ket{z}$ can be done in $\calO(k d^2)$ time and $\calO(d^2)$ space using common matrix operations over $\CO^d$. Finally, for any vector $\ket{\hat{z}} \in \CO^d$, vector $L^\dagger \ket{\hat{z}}$ can be represented as a decision diagram of size $\calO(s)$ by performing $d$ scalar multiplications and additions over decision diagrams of size $\calO(s)$, yielding a time and space complexity $\calO(s d)$.
\end{proof}

Provided that the sequence $U^1 \ket{z}, U^2 \ket{z}, \ldots$ can be computed using decision diagrams whose size is polynomial in the number of qubits, Theorem~\ref{thm:dd} ensures that the quantum lumping of $U$ w.r.t. $S_{\ket{z}}$ can be computed with time and space requirements that are polynomial in the number of qubits. 

At the expense of higher complexity, Theorem~\ref{thm:dd} can be extended to quantum lumpings w.r.t. arbitrary subspaces, as stated next.

\begin{corollary}\label{cor:dd}
Fix a quantum circuit $C$ that induces the unitary map $U \in \CO^{N \times N}$ and let $S \subseteq \CO^N$ be a subspace spanned by $\ket{z_1},\ldots,\ket{z_\kappa} \in \CO^N$. Assume that each quantum lumping $L_l$ of $U$ w.r.t. $S_{\ket{z_l}}$ has dimension $d_l$ and assume that all $U^1 \ket{z_l}, \ldots, U^{d_l} \ket{z_l}$ can be computed and stored as decision diagrams in time and space $\calO(s)$ using $C$. Then, 1. and 2. of Theorem~\ref{thm:dd} carry over for $d = \sum_l d_l$.
\end{corollary}
\begin{proof}
We begin with statement 1. We apply the procedure from Theorem~\ref{thm:dd} to each $\ket{z_l}$ in isolation, thus computing a quantum lumping $L_l$ of $U$ wrt $S_{\ket{z_l}}$. The overall time and space complexity is $\calO(\sum_l s d_l^2)$ and $\calO(\sum_l s d_l)$, respectively. Afterwards, we apply MGS to the union of all columns $L_1^\dagger, \ldots, L_\kappa^\dagger$. This comes with time complexity $\calO(s (\sum_l d_l)^2)$ and space complexity $\calO(\sum_l s d_l)$. Statement 2. is shown as in Theorem~\ref{thm:dd}.
\end{proof}

\section{Applications}\label{sec:applications}

In this section, we demonstrate that established quantum algorithms enjoy substantial bisimulation reductions. For each application, we provide a brief description of the quantum algorithm and a theoretical bound on its reduction.

\subsection{Quantum Search}\label{sec:qsearch}

Let us assume that we are given a non-zero function $f : \{0,1\}^n \to \{0,1\}$ and that we are asked to find some $x \in \{0,1\}^n$ such that $f(x) = 1$. Grover's seminal algorithm describes how this can be achieved in $\mathcal{O}(\sqrt{N})$ steps on a quantum computer~\cite[Section 6.2]{nielsen00}, thus yielding a quadratic speed-up over a classic computer. For any $x \in \{0,1\}^n$, the Grover map is given by

\begin{align}\label{equ:grover}
G \ket{x} = (-1)^{f(x)} (I - 2 \ket{\psi} \bra{\psi}) \ket{x} , \quad \text{with} \ \
\ket{\psi} = \frac{1}{\sqrt{N}} \sum_{x = 0}^{N-1} \ket{x}
\end{align}
The Grover map yields the following celebrated result.

\begin{theorem}[Quantum Search~\cite{nielsen00}]\label{thm:grover}
Map $G$ is unitary. Moreover, if the number of sought solutions $M = |\{ x \mid f(x) = 1 \}|$ satisfies $M \leq N/2$, then measuring $G^\kappa \ket{\psi}$ for $\kappa = \lceil \frac{\pi}{4} \sqrt{N/M} \rceil$ yields a state $\ket{x}$ satisfying $f(x) = 1$ with probability at least $\frac{1}{2}$.
\end{theorem}

The next result allows one to compute the result $G^\kappa \ket{\psi}$ from Theorem~\ref{thm:grover} using a map over a single qubit.

\begin{theorem}[Reduced Grover]\label{thm:grover:lump}
The BCB $L^\dagger \in \CO^{N \times d}$ of $G$ w.r.t. $S_{\ket{\psi}}$ has dimension $d = 2$ and a column space spanned by $\ket{\psi}$ and $G \ket{\psi}$.
\end{theorem}

\begin{proof}
The claim follows by noting that the column space of an BCB w.r.t. $S_{\ket{\psi}}$ is spanned by $\ket{\psi}, G \ket{\psi}, G^2 \ket{\psi}, \ldots, G^{N-1} \ket{\psi}$ and so on. This, in turn, is known to have as basis~\cite[Section 6.2]{nielsen00}
        \begin{align*}
        \ket{\alpha} & = \frac{1}{\sqrt{M}} \sum_{x : f(x) = 1} \ket{x} & \text{and} & &
        \ket{\beta} & = \frac{1}{\sqrt{N-M}} \sum_{x : f(x) = 0} \ket{x} ,
        \end{align*}
        where $M$ is as above, while $\ket{\alpha}$ is the superposition (i.e., sum) of all solution states and $\ket{\beta}$ is the superposition of all non-solution states.
\end{proof}

\begin{remark}
Although BCB $L^\dagger$ wrt $S_{\ket{\psi}}$ always has dimension $2$, its column space depends on the oracle function $f$. This is because $f$ appears in $G$, see~(\ref{equ:grover}).
\end{remark}

\subsection{Quantum Optimization}\label{sec:qaoa}

The Quantum Approximation Optimization Algorithm (QAOA)~\cite{farhi2014quantum} is a computational model that has the same expressive power as the common quantum circuit model~\cite{farhi2000quantum,farhi2014quantum,1366223}. It is described by two matrices. The first is the \emph{problem Hamiltonian} $H_P$ for which we are interested in computing a maximal eigenstate, i.e., an eigenvector for a maximal eigenvalue of $H_P$. The second is the  \emph{begin Hamiltonian} $H_B$ for which a maximal eigenstate $\ket{\psi}$ is known already. With this, a maximal eigenstate of $H_P$ can be obtained by performing the QAOA introduced next. 

\begin{definition}[QAOA~\cite{farhi2014quantum}] \label{def:qaoa}
For a problem Hamiltonian $H_P$ and a begin Hamiltonian $H_B$, fix the unitary matrices
\begin{align*}
U_B(\delta) & = \exp(-i \delta H_B) & \text{and} &  &
U_P(\delta) & = \exp(-i \delta H_P)
\end{align*}
where $\delta > 0$ is a sufficiently small time step and $\exp(A)$ is the matrix exponential. For a sequence of natural numbers $(k_i,l_i)_{i=1}^\kappa$ of length $\kappa \geq 1$, we define

\begin{equation}\label{eq:qaoa:expanded}
        \ket{w_\kappa} = U_B(\delta)^{k_\kappa} U_P(\delta)^{l_\kappa} \cdot \ldots \cdot U_B(\delta)^{k_1} U_P(\delta)^{l_1} \ket{\psi}
    \end{equation}
The QAOA with $\kappa \geq 1$ stages is then given by $\max \{ \bra{w_\kappa} H_P \ket{w_\kappa} \mid (k_i,l_i)_{i=1}^\kappa \}$. \end{definition}

While the problem Hamiltonian $H_P$ depends on the task or problem we are solving, the choice of the begin Hamiltonian $H_B$ is informed by the so-called adiabatic theorem, a result that identifies conditions under which QAOA returns a global optimum. A common heuristic is to pick $H_B$ such that $H_B$ and $H_P$ do not diagonalize over a common basis~\cite{farhi2000quantum,farhi2014quantum}  and to assume without loss of generality that $\ket{\psi} =  \sum_x \ket{x} / \sqrt{N}$ is the unique maximal eigenvector of $H_B$.

We next demonstrate that bisimulation can reduce QAOA when applied to SAT and MaxCut, two NP-complete problems~\cite{sipser1996introduction}. We start by introducing the problem Hamiltonians $H_P$ for both cases.

\begin{definition}[SAT and MaxCut Problem Hamiltonians]\label{def:ham:sat}
\begin{itemize}
    \item For a boolean formula $\phi = \bigwedge_{i = 1}^M C_i$, where $C_i$ is a clause over $n$ boolean variables, the problem Hamiltonian is given by $H_P = \sum_i H_i$, where
        \[
        H_i \ket{x} =
        \begin{cases}
        \ket{x} & , \ C_i(x) \text{ is true} \\
        0 & , \ C_i(x) \text{ is false} \\
        \end{cases}
        \]
        for any $x \in \{0,1\}^n$ that represents a Boolean assignment.
    \item For an undirected unweighted graph $G = (V,E)$ with vertices $V = \{1,\ldots,n\}$ and edges $E \subseteq V \times V$, we define the problem Hamiltonian $H_P = \sum_{(i,j) \in E} H_{i,j}$, where
        \[
        H_{i,j} \ket{x} =
        \begin{cases}
        \ket{x} & , \ x_i \neq x_j  \\
        0 & , \ x_i = x_j \\
        \end{cases}
        \]
        for any $x \in \{0,1\}^n$ that represents a cut $C \subseteq \{1,\ldots,n\}$ by setting $i \in C$ if and only if $x_{i-1} = 1$.
\end{itemize}
\end{definition}

Following this definition, it can be shown that the QAOA $\bra{w_\kappa} H_P \ket{w_\kappa}$ from Definition~\ref{def:qaoa} corresponds to a quantum measurement reporting either the expected number of satisfied clauses or the expected size of the cut. 
It is possible to guarantee that QAOA finds a global optimum for a sufficiently high $\kappa$~\cite{farhi2000quantum,farhi2014quantum}.

The next result ensures that $H_P$ has BCB $L^\dagger$ w.r.t. $S_{\ket{\psi}}$ whose reduced map is provably small. Moreover, for any such $L$, it ensures that there exists a begin Hamiltonian $H_B$ for which $L^\dagger$ is a BCB too, thus ensuring substantial reductions of the entire QAOA calculation~(\ref{eq:qaoa:expanded}).

\begin{theorem}[Reduced QAOA]\label{thm:ham:sat}
Fix $H_P$ as in Definition~\ref{def:ham:sat}, any $\delta > 0$ and let $L^\dagger \in \CO^{N \times d}$ be a BCB of $U_P(\delta)$ w.r.t. $S_{\ket{\psi}}$. Then
\begin{enumerate}
    \item The column space of $L^\dagger$ is spanned by
        \begin{align}\label{eq:ham:sat}
        \big(\ket{\psi}, U_P(\delta) \ket{\psi}, U^2_P(\delta) \ket{\psi}, \ldots, U^{M-1}_P(\delta)\ket{\psi} \big)^\dagger ,
        \end{align}

    where $M$ is the number of clauses (SAT) or edges (MaxCUT). Specifically, the dimension of the BCB $d$ is bounded by $M$.

    \item Then, for any Hamiltonian $\hat{H}_B \in \CO^{d \times d}$ (i.e., Hermitian matrix), there is a Hamiltonian $H_B \in \CO^{N \times N}$ such that
        \begin{itemize}
            \item $L^\dagger$ is a BCB of $U_B(\delta) = \exp(-i \delta H_B)$ wrt $S_{\ket{\psi}}$, while its reduced map is $\hU_B(\delta) = \exp \, (-i \delta \hat{H}_B)$

            \item The computation~(\ref{eq:qaoa:expanded}) satisfies
            \begin{align}
            \ket{w_\kappa} & = U_B(\delta)^{k_\kappa} U_P(\delta)^{l_\kappa} \cdot \ldots \cdot U_B(\delta)^{k_1} U_P(\delta)^{l_1} \ket{\psi} \nonumber \\
            & = L^\dagger \hU_B(\delta)^{k_\kappa} \hU_P(\delta)^{l_\kappa} \cdot \ldots \cdot \hU_B(\delta)^{k_1} \hU_P(\delta)^{l_1} L \ket{\psi} \label{eq:qaoa:red}
            \end{align}
            The QAOA in $\CO^N$ thus corresponds to a QAOA in the reduced space $\CO^d$.
            \end{itemize}
    \end{enumerate}
\end{theorem}

\begin{proof}
We begin by proving 1. For SAT, it can be seen that $H_P \ket{x} = \nu \ket{x}$, where $0 \leq \nu \leq M$ is the number of clauses that are satisfied by assignment $x$. A similar formula holds for MaxCut, the difference being that $\nu$ is the size of the cut $x$. 
It should be noted that $H_P$ is in diagonal form for both SAT and MaxCut. If $m$ denotes the number of distinct eigenvalues of $H_P$, then $m \leq M$, where $M$ is in the case of SAT or MaxCUT, respectively, the number of clauses or edges. The same can be said regarding its matrix exponential $U_P(\delta)$ which, being unitary, enjoys an eigendecomposition, allowing us to write $\ket{\psi} = \sum_{i = 1}^{m} c_i \ket{z_i}$, where $\ket{z_i}$ is an eigenvector for eigenvalue $\lambda_i$ of $U_P(\delta)$. This yields
\[
U^k \ket{\psi} = \sum_{i = 1}^{m} c_i \lambda_i^k \ket{z_i}
\]
for all $k \geq 0$. Without loss of generality, consider $d \leq m$ such that $c_k = 0$ for all $k > d$ and $c_k \neq 0$ otherwise. Writing the vectors $\{ U^k \ket{\psi} \mid  0 \leq k \leq m-1 \}$ in the basis $\ket{z_1}, \ldots, \ket{z_{d}}$ gives rise to a regular Vandermonde matrix~\cite{DMD2009} in $\CO^{d \times d}$. This shows that $\{ U^k \ket{\psi} \mid  d \leq k \leq M-1 \}$ are linear combinations of $\{ U^k \ket{\psi} \mid  0 \leq k \leq d-1 \}$, completing the proof of 1. Instead, 2. follows from the definition of BCB and Lemma~\ref{lem:ham:sat} from below.
\end{proof}

The following auxiliary result is needed in the proof of Theorem~\ref{thm:ham:sat}.

\begin{lemma}\label{lem:ham:sat}
Pick any $L \in \CO^{d \times N}$ and $Q \in \CO^{(N - d) \times N}$ so that the rows of $L$ and $Q$ comprise an orthonormal basis of $\CO^N$, and define
\begin{align*}
    U_B & = L^\dagger \hU_B L + Q^\dagger \tilde{U}_B Q, &
    \hU_B & = \exp(- i \delta \hat{H}_B), &
    \tilde{U}_B & = \exp(- i \delta \tilde{H}_B)
\end{align*}
for any Hamiltonian $\hat{H}_B \in \CO^{d \times d}$ and $\tilde{H}_B \in \CO^{(N-d) \times (N - d)}$. Then, $U_B$ is unitary, $L$ is an FCB of it wrt $S_{\ket{\psi}}$, and $\hU_B$ is its reduced map. In addition, there exists a Hamiltonian $H_B \in \CO^{N \times N}$ that satisfies $U_B = \exp(- i \delta H_B)$.
\end{lemma}

\begin{proof}
We first show that $L U_B = L U_B L^\dagger L$ as this implies that $L$ is an FCB of $U$ by~\cite{tomlin1997effect}. To see this, we note that
\begin{align*}
    L U_B L^\dagger L & = L \big( L^\dagger \hU_B L + Q^\dagger \tilde{U}_B Q \big) L^\dagger L
    = L L^\dagger \hU_B L L^\dagger L +  L Q^\dagger \tilde{U}_B Q L^\dagger L
    = \hU_B L \\
    L U_B & = L \big( L^\dagger \hU_B L + Q^\dagger \tilde{U}_B Q \big)
    = L L^\dagger \hU_B L + L Q^\dagger \tilde{U}_B Q
    = \hU_B L
\end{align*}
where we have used that $L L^\dagger = 0$ and $L Q^\dagger = 0$, which follows from the choice of $Q$. From the calculation, we can also infer that $\hU_B = L U_B L^\dagger$, i.e., $\hU_B$ is indeed the reduced map. In a similar fashion, one can note that $Q$ is also an FCB of $U_B$ and that $\tilde{U}_B$ is the respective reduced map. Since both $\hU_B$ and $\tilde{U}_B$ are unitary, we infer that also $U_B$ is unitary (alternatively, a direct calculation yields $I = U_B U_B^\dagger$). Since any unitary matrix can be written as a matrix exponential of a Hamiltonian, there exists a Hamiltonian $H_B$ satisfying $U_B = \exp(-i\delta H_B)$.
\end{proof}

\subsection{Quantum Factorization and Order Finding}\label{sec:qorder}

Let us assume that we are given a composite number $N$ which we seek to factorize. As argued in~\cite[Section 5.3.2]{nielsen00}, this problem can be solved in randomized polynomial time, provided that the same holds for the order finding problem. Given some randomly chosen $x \in \{2,3\ldots,N-1\}$, the latter asks to compute the multiplicative order of $x$ modulo $N$, i.e., the smallest $r \geq 1$ satisfying $x^r \mod N = 1$. Following~\cite[Section 5.3.1]{nielsen00}, we consider the quantum algorithm defined by the unitary map
\[
U \ket{y} =
\begin{cases}
\ket{ xy \mod N } & , \ 0 \leq y < N \\
\ket{y} & , \ N \leq y < 2^l
\end{cases}
\]
Here, $l \geq 1$ is the smallest number satisfying $N \leq 2^l$.

The next result allows us to relate the order of $x$ to the dimension of the BCB wrt $S_{\ket{1}}$. This fact is exploited in Shor's factorization algorithm~\cite{nielsen00}.

\begin{theorem}[Reduced Order Finding]\label{thm:u:factor:lump}
The dimension of the BCB of $U$ wrt $S_{\ket{1}}$ is equal to the order of $x$ modulo $N$.
\end{theorem}

\begin{proof}
If can be shown~\cite{nielsen00} that the $U$ from above is unitary and that $U \ket{u_s} = e^{2\pi i s / r} \ket{u_s}$ for all $0 \leq s \leq r - 1$, where
\begin{align*}
\ket{u_s} & = \frac{1}{\sqrt{r}} \sum_{k = 0}^{r - 1} \exp\big[ \frac{-2\pi i s k}{r} \big] \ket{x^k \mod N} &
& \text{and} &
\frac{1}{\sqrt{r}}\sum_{s = 0}^{r - 1} \ket{u_s} & = \ket{1} .
\end{align*}
With this, $U^k \ket{1} = \frac{1}{\sqrt{r}} \sum_{s = 0}^{r-1} (e^{2\pi i s / r})^k u_s$ for any $p \geq 0$. Hence, the minimal BCB with respect to $S_{\ket{1}}$ is contained in the span of $u_0,\ldots,u_{r-1}$. To see that the dimension is exactly $r$, we note that the vectors $\{ U^k \ket{\psi} \mid  0 \leq k \leq r - 1 \}$, written in the basis $u_0, \ldots, u_{r-1}$, constitute a regular Vandermonde matrix~\cite{DMD2009} in $\CO^{r \times r}$.
\end{proof}

\section{Numerical Experiments}\label{sec:benchmark}

We evaluated our approach on the applications from Section~\ref{sec:applications} and the MQT Bench quantum circuit benchmark suite~\cite{quetschlich2022mqtbench}. To this end, we extended the publicly available implementation of the CLUE algorithm from~\cite{ovchinnikov2021,leguizamon2023} which used Python, Gaussian elimination and sparse vectors and matrices. The new version was re-implemented in C++ and uses the state-of-the-art decision diagram package for quantum computing provided as part of the Munich Quantum Toolkit's MQT Core library (available at \url{https://github.com/cda-tum/mqt-core}). A version that uses sparse vectors and matrices for representing state vectors and quantum circuits, respectively, is offered for completeness. The respective prototype is publicly accessible at the GitHub repository \texttt{CLUE}, branch \texttt{quantum+cpp}\footnote{\url{https://github.com/Antonio-JP/CLUE/tree/quantum\%2Bcpp}}, all results reported were executed on a machine with an i7-8665U CPU, 32GB RAM and a 1024GB SSD. 

As anticipated, the evaluation demonstrates that a combination of quantum bisimulation with decision diagrams outperforms each individual approach.



\subsection{Applications from Section~\ref{sec:applications}}\label{subsec:exp_clue}

Here we report the results of numerical experiments on the applications discussed in Section~\ref{sec:applications}. Starting with five qubits, we increased the number until a timeout of $500$ seconds has been reached. 
To allow for a representative evaluation, we averaged runtimes over $50$ independent runs for each case study. Specifically, the instances were generated as follows:
\begin{itemize}
    \item \emph{Grover algorithm} (Sec.~\ref{sec:qsearch}): for a fixed numbers $n$ of qubits, we pick randomly an element $m \in\{0,\ldots,2^n-1\}$. We create the corresponding circuit for Grover's algorithms with $n$ qubits (1 of them is ancillary) that looks for the element $m$.
    \item \emph{Quantum Optimization for SAT} (Sec.~\ref{sec:qaoa}): for each number of qubits $n$, we generate a random formula with $m$ clauses ($m$ is randomly selected between $n$ and $3n$), where each clause has 3 variables at most. We guarantee that every formula contains all $n$ variables.
    \item \emph{Quantum Optimization for MaxCut} (Sec.~\ref{sec:qaoa}): for each number of qubits $n$, we generate an Erd\H{o}s-R\'enyi graph with $n$ nodes and edge probability $\tfrac{1}{3}$.
\end{itemize}


We considered three different regimes: 
\begin{itemize}
    \item CLUE:  quantum bisimulation is computed using CLUE that relies on vector matrix representations; the simulation is conducted using the reduced system; 
    \item DDSIM+CLUE: quantum bisimulation is computed using CLUE that relies on decision diagrams; the simulation is conducted using the reduced system; 
    \item DDSIM: using decision diagrams, we simulate the full circuit. 
\end{itemize}
 The number of steps $\kappa$ was set to the smallest integer greater than or equal to $\sqrt{N}$. The choice of $\kappa$ is motivated by Theorem~\ref{thm:grover} and the discussion around the so-called adiabatic theorem in~\cite{farhi2000quantum,farhi2014quantum}.


\begin{table*}[tp]
\begin{adjustbox}{center}
    \scriptsize
    \begin{tabular}{c|rrr|rrr|rrr}
\toprule
 & \multicolumn{3}{c}{\textbf{Grover}} & \multicolumn{3}{c}{\textbf{SAT}} & \multicolumn{3}{c}{\textbf{MAXCUT}} \\
 & \textbf{CLUE} & \textbf{DDSIM+CLUE} & \textbf{DDSIM} & \textbf{CLUE} & \textbf{DDSIM+CLUE} & \textbf{DDSIM} & \textbf{CLUE} & \textbf{DDSIM+CLUE} & \textbf{DDSIM} \\
\midrule
5 & 0.0028 & 0.0020 & \bfseries 0.0019 & \bfseries 0.0035 & 0.0111 & 0.0134 & \bfseries 0.0022 & 0.0038 & 0.0070 \\
6 & 0.0112 & \bfseries 0.0023 & 0.0025 & \bfseries 0.0080 & 0.0160 & 0.0758 & \bfseries 0.0069 & 0.0086 & 0.0339 \\
7 & 0.0463 & \bfseries 0.0026 & 0.0032 & \bfseries 0.0186 & 0.0326 & 0.4515 & 0.0225 & \bfseries 0.0179 & 0.2372 \\
8 & 0.1752 & \bfseries 0.0032 & 0.0051 & \bfseries 0.0460 & 0.0746 & 1.8876 & 0.0593 & \bfseries 0.0421 & 1.9653 \\
9 & 0.6328 & \bfseries 0.0037 & 0.0072 & \bfseries 0.1150 & 0.1493 & 13.5664 & 0.1606 & \bfseries 0.1173 & 17.1287 \\
10 & 1.8453 & \bfseries 0.0049 & 0.0115 & \bfseries 0.1789 & 0.2560 & 46.2972 & 0.3117 & \bfseries 0.1675 & 52.3542 \\
11 & 7.1148 & \bfseries 0.0069 & 0.0174 & 0.4323 & \bfseries 0.3356 & 208.8560 & 0.8154 & \bfseries 0.3567 & 313.6538 \\
12 & 29.8714 & \bfseries 0.0114 & 0.0274 & 1.0144 & \bfseries 0.6830 & 271.8651 & 2.4235 & \bfseries 0.7538 & 418.2120 \\
13 & 128.0926 & \bfseries 0.0196 & 0.0431 & 2.6211 & \bfseries 1.3712 & >500 & 7.3718 & \bfseries 2.3766 & >500 \\
14 & >500 & \bfseries 0.0349 & 0.0707 & 5.6021 & \bfseries 2.4126 & >500 & 17.3324 & \bfseries 4.2761 & >500 \\
15 & >500 & \bfseries 0.0670 & 0.1077 & 13.2476 & \bfseries 4.9066 & >500 & 43.4601 & \bfseries 9.9412 & >500 \\
16 & >500 & \bfseries 0.1231 & 0.1801 & 30.1186 & \bfseries 10.0833 & >500 & 111.2968 & \bfseries 48.0762 & >500 \\
17 & >500 & \bfseries 0.2470 & 0.2976 & 64.4899 & \bfseries 26.5241 & >500 & 279.0988 & \bfseries 135.6231 & >500 \\
18 & >500 & \bfseries 0.4723 & 0.4993 & 157.1778 & \bfseries 89.1793 & >500 & >500 & \bfseries 334.3610 & >500 \\
19 & >500 & 0.9447 & \bfseries 0.9167 & 298.1822 & \bfseries 113.3666 & >500 & >500 & >500 & >500 \\
20 & >500 & 1.8267 & \bfseries 1.7992 & >500 & \bfseries 125.5865 & >500 & >500 & >500 & >500 \\
21 & >500 & 3.5860 & \bfseries 3.4293 & >500 & >500 & >500 & >500 & >500 & >500  \\
22 & >500 & \bfseries 7.2422 & 7.7943 & >500 & >500 & >500 & >500 & >500 & >500  \\
23 & >500 & \bfseries 14.4505 & 17.1884 & >500 & >500 & >500 & >500 & >500 & >500 \\
24 & >500 & \bfseries 28.4858 & 45.6023 & >500 & >500 & >500 & >500 & >500 & >500 \\
25 & >500 & \bfseries 58.6463 & 142.5514 & >500 & >500 & >500 & >500 & >500 & >500 \\
26 & >500 & \bfseries 112.1946 & 331.1369 & >500 & >500 & >500 & >500 & >500 & >500 \\
27 & >500 & \bfseries 225.3220 & >500 & >500 & >500 & >500 & >500 & >500 & >500 \\
28 & >500 & \bfseries 455.1170 & >500 & >500 & >500 & >500 & >500 & >500 & >500 \\
\bottomrule
\end{tabular}
\end{adjustbox}
\caption{Simulation times of quantum applications from Section~\ref{sec:applications} in case of three different regimes. CLUE: the circuit is reduced using vectors and simulated; DDSIM+CLUE: the circuit is reduced using decision diagrams and then simulated; DDSIM: the original circuit is simulated using decision diagrams without any reduction.
}\label{tab:examples}
\end{table*}

For quantum optimization, Table~\ref{tab:examples} reports logarithmic CLUE reductions, reducing in particular $2^{15} \times 2^{15}$ matrices to $15 \times 15$ matrices in less than $15$\,s. Moreover, we observe that a combination of CLUE with DDSIM outperforms each individual approach.

\subsection{Benchmark Circuits}\label{sec:eval}

Similarly to Table~\ref{tab:examples}, we next compare the scalability of vector-matrix operations versus that of decision diagrams, this time on quantum benchmarks from~\cite{quetschlich2022mqtbench}, available at \url{https://www.cda.cit.tum.de/mqtbench/}. To this end, we computed the constrained bisimulation w.r.t. $\ket{0}$ for all circuits in each benchmark family until a timeout of 500 seconds has been reached. We show the reduction ratio (RR) for the constrained bisimulation w.r.t. $\ket{0}$ and the times (in seconds) spent under regimes CLUE and DDSIM+CLUE, averaged over 5 executions. We note that some circuits were only available for a specific number of qubits (e.g., HHL and price calls). For the benefit of presentation, we omit qubit numbers whose simulation times were below one second for both CLUE and DDSIM+CLUE.

In general, it can be observed that substantial reductions could be obtained for a number of benchmark families. Concerning scalability, we note that DDSIM+CLUE outperforms CLUE in most cases. An exception to this rule is when a circuit cannot be reduced. In such a case, an extension of CLUE by decision diagrams tends to perform worse.


\begin{table}[tp]
\scriptsize
\caption{Reduction Ratio (RR) and Time (s) comparison between CLUE and DDSIM+CLUE for computing \
                        a lumping w.r.t. $\ket{0}$.}
\label{table:benchmark}
\begin{adjustbox}{center}
\begin{tabular}[t]{rc|rr|rrH|}
\toprule
 &  & \multicolumn{2}{c}{\textbf{CLUE}} & \multicolumn{2}{c}{\textbf{DDSIM+CLUE}} & \textbf{DDSIM} \\
 &  & \textbf{RR for $S_{\ket{0}}$} & \textbf{Time (s)} & \textbf{RR for $S_{\ket{0}}$} & \textbf{Time (s)} & \textbf{Time (s)} \\
\midrule
\multirow[c]{3}{*}{ae} & 7 & 100.00\% & 4.0022 & 100.00\% & \bfseries 3.4999 & 0.0030 \\
 & 8 & 100.00\% & \bfseries 31.7723 & 100.00\% & 32.5920 & 0.0061 \\
 & 9 & 100.00\% & \bfseries 264.9636 & - & >500 & 0.0166 \\
\cline{1-7}
\multirow[c]{6}{*}{dj} & 15 & 0.01\% & 0.1223 & 0.01\% & \bfseries 0.0252 & 0.0015 \\
 & 16 & - & >500 & <0.01\% & \bfseries 0.0499 & 0.0017 \\
 & \vdots & \vdots & \vdots & \vdots & \vdots & \vdots\\ 
 & 27 & - & >500 & <0.01\% & \bfseries 65.7853 & 0.0029 \\
 & 28 & - & >500 & <0.01\% & \bfseries 125.6654 & 0.0036 \\
 & 29 & - & >500 & <0.01\% & \bfseries 249.1532 & 0.0033 \\
\cline{1-7}
\multirow[c]{5}{*}{ghz} & 11 & 1.56\% & 2.5651 & 1.56\% & \bfseries 2.3672 & 0.0003 \\
 & 12 & 0.78\% & 5.0972 & 0.78\% & \bfseries 3.5408 & 0.0003 \\
 & 13 & 0.39\% & \bfseries 10.2004 & 0.39\% & 11.8914 & 0.0003 \\
 & 14 & 0.20\% & \bfseries 14.9587 & 0.20\% & 40.8600 & 0.0004 \\
 & 15 & 0.10\% & \bfseries 30.5938 & - & >500 & 0.0003 \\
\cline{1-7}
\multirow[c]{7}{*}{graphstate} & 8 & 9.38\% & 1.6081 & 9.38\% & \bfseries 0.1426 & 0.0002 \\
 & 9 & 3.52\% & 6.2612 & 3.52\% & \bfseries 0.2846 & 0.0004 \\
 & 10 & 0.98\% & 13.0928 & 0.98\% & \bfseries 0.1867 & 0.0003 \\
 & 11 & 2.34\% & 202.4504 & 2.34\% & \bfseries 2.5678 & 0.0004 \\
 & 12 & 0.29\% & 210.7506 & 0.29\% & \bfseries 0.9886 & 0.0004 \\
 & 13 & - & >500 & 0.15\% & \bfseries 0.8902 & 0.0004 \\
 & 14 & - & >500 & 0.13\% & \bfseries 33.3943 & 0.0005 \\
\cline{1-7}
\multirow[c]{1}{*}{hhl} & 8 & 87.50\% & \bfseries 12.4598 & 39.30\% & 18.8924 & 0.0674 \\
\cline{1-7}
\multirow[c]{4}{*}{portfolioqaoa} & 6 & 100.00\% & \bfseries 0.5800 & 100.00\% & 1.1288 & 0.0049 \\
 & 7 & 100.00\% & \bfseries 4.1406 & 100.00\% & 5.8476 & 0.0096 \\
 & 8 & 100.00\% & \bfseries 31.2608 & 100.00\% & 54.4699 & 0.0181 \\
 & 9 & 100.00\% & \bfseries 263.1970 & - & >500 & 0.0549 \\
\cline{1-7}
\multirow[c]{3}{*}{portfoliovqe} & 7 & 100.00\% & \bfseries 3.9064 & 100.00\% & 4.0044 & 0.0052 \\
 & 8 & 100.00\% & 31.6000 & 100.00\% & \bfseries 30.6770 & 0.0091 \\
 & 9 & 100.00\% & \bfseries 266.4854 & 100.00\% & 490.4663 & 0.0211 \\
\cline{1-7}
\multirow[c]{3}{*}{pricingcall} & 11 & 3.17\% & 8.6745 & 3.17\% & \bfseries 1.9269 & 0.0068 \\
 & 13 & 1.57\% & 171.3546 & 1.57\% & \bfseries 16.7369 & 0.0162 \\
 & 15 & - & >500 & 0.79\% & \bfseries 142.1170 & 0.0401 \\
\cline{1-7}
\multirow[c]{3}{*}{pricingput} & 11 & 3.17\% & 7.3885 & 3.17\% & \bfseries 1.4735 & 0.0057 \\
 & 13 & 1.57\% & 129.0688 & 1.57\% & \bfseries 10.8187 & 0.0156 \\
 & 15 & - & >500 & 0.79\% & \bfseries 130.6928 & 0.0429 \\
\cline{1-7}
\end{tabular}%
\begin{tabular}[t]{|rc|rr|rrH}
\toprule
 &  & \multicolumn{2}{c}{\textbf{CLUE}} & \multicolumn{2}{c}{\textbf{DDSIM+CLUE}} & \textbf{DDSIM} \\
 &  & \textbf{RR for $S_{\ket{0}}$} & \textbf{Time (s)} & \textbf{RR for $S_{\ket{0}}$} & \textbf{Time (s)} & \textbf{Time (s)} \\
\midrule
\multirow[c]{8}{*}{qft} & 10 & 0.20\% & 3.0963 & 0.20\% & \bfseries 0.0041 & 0.0009 \\
 & 11 & 0.10\% & 11.0858 & 0.10\% & \bfseries 0.0061 & 0.0011 \\
 & 12 & 0.05\% & 45.4641 & 0.05\% & \bfseries 0.0086 & 0.0013 \\
 & 13 & - & >500 & 0.02\% & \bfseries 0.0136 & 0.0015 \\
 & \vdots & \vdots & \vdots & \vdots & \vdots & \vdots\\ 
 & 25 & - & >500 & <0.01\% & \bfseries 33.9910 & 0.0055 \\
 & 26 & - & >500 & <0.01\% & \bfseries 69.0587 & 0.0061 \\
 & 27 & - & >500 & <0.01\% & \bfseries 171.3884 & 0.0067 \\
\cline{1-7}
\multirow[c]{4}{*}{qnn} & 6 & 100.00\% & \bfseries 0.5825 & 100.00\% & 1.6813 & 0.0077 \\
 & 7 & 100.00\% & \bfseries 4.0485 & 100.00\% & 6.6028 & 0.0098 \\
 & 8 & 100.00\% & \bfseries 31.4335 & 100.00\% & 52.0715 & 0.0194 \\
 & 9 & 100.00\% & \bfseries 263.1350 & - & >500 & 0.0317 \\
\cline{1-7}
\multirow[c]{4}{*}{qpeexact} & 7 & 100.00\% & 1.3285 & 50.00\% & \bfseries 0.3797 & 0.0004 \\
 & 8 & 100.00\% & 10.8579 & 50.00\% & \bfseries 2.3920 & 0.0005 \\
 & 9 & 100.00\% & 88.0831 & 50.00\% & \bfseries 28.7468 & 0.0006 \\
 & 10 & - & >500 & 50.00\% & \bfseries 472.6555 & 0.0010 \\
\cline{1-7}
\multirow[c]{3}{*}{qpeinexact} & 7 & 100.00\% & 1.7021 & 50.00\% & \bfseries 0.5057 & 0.0013 \\
 & 8 & 100.00\% & 13.1596 & 50.00\% & \bfseries 2.9213 & 0.0018 \\
 & 9 & 100.00\% & 107.9556 & 50.00\% & \bfseries 30.2104 & 0.0043 \\
\cline{1-7}
\multirow[c]{5}{*}{qwalk} & 6 & 50.00\% & \bfseries 0.0466 & 50.00\% & 1.1475 & 0.0097 \\
 & 7 & 62.50\% & \bfseries 0.3807 & 50.00\% & 4.5801 & 0.0165 \\
 & 8 & 100.00\% & \bfseries 6.9250 & 50.00\% & 30.7021 & 0.0285 \\
 & 9 & 100.00\% & \bfseries 56.6530 & - & >500 & 0.0539 \\
 & 10 & 100.00\% & \bfseries 465.2573 & - & >500 & 0.1178 \\
\cline{1-7}
tsp & 9 & 41.41\% & 74.2882 & 41.41\% & \bfseries 72.8202 & 0.0163 \\
\cline{1-7}
\multirow[c]{4}{*}{vqe} & 6 & 100.00\% & \bfseries 0.5768 & 100.00\% & 4.5901 & 0.0019 \\
 & 7 & 100.00\% & \bfseries 3.8829 & 100.00\% & 8.2355 & 0.0044 \\
 & 8 & 100.00\% & \bfseries 31.1499 & 100.00\% & 457.2644 & 0.0050 \\
 & 9 & 100.00\% & \bfseries 264.7780 & 100.00\% & 414.3664 & 0.0100 \\
\cline{1-7}
\multirow[c]{3}{*}{wstate} & 7 & 100.00\% & \bfseries 2.3740 & 100.00\% & 2.5267 & 0.0003 \\
 & 8 & 100.00\% & \bfseries 18.2656 & 100.00\% & 23.3446 & 0.0004 \\
 & 9 & 100.00\% & \bfseries 151.7882 & 100.00\% & 391.4060 & 0.0004 \\
\cline{1-7}
\bottomrule
\end{tabular}
\end{adjustbox}
\end{table}




\section{Conclusion}\label{sec:conc}

We introduced forward and backward constrained bisimulations for quantum circuits which allow by means of reduction to preserve an invariant subspace of interest. The applicability of the approach was demonstrated by obtaining substantial reductions of common quantum algorithms, including, in particular, quantum search, quantum approximate optimization algorithms for SAT and MaxCut, as well as a number of benchmark circuits. Overall, the results suggest that constrained bisimulation can be used as a tool for speeding up the simulation of quantum circuits on classic computers, especially when the circuit is to be simulated under several initial conditions or for multi-step applications.

In line with the relevant literature on bisimulations for dynamical systems, constrained bisimulations introduce loss of information due to their underlying projection onto a smaller-dimensional state space; the information that is preserved, however, is exact. A relevant issue for future work is to consider approximate variants of bisimulation for quantum circuits, in order to find more aggressive reductions or to capture quantum-specific phenomena such as quantum noise.


\bibliographystyle{ACM-Reference-Format}
\bibliography{bibliography}

\appendix

\end{document}